\documentclass[letterpaper, 10 pt, conference]{ieeeconf}

\IEEEoverridecommandlockouts   
\usepackage{cite}

\usepackage{epsfig} 
\usepackage{graphics}
\usepackage{amssymb}
\usepackage{amsmath, amsthm}
\usepackage{xcolor}

\newcommand{\Real}{\mathbb{R}}
\newcommand{\N}{\mathbb{N}}
\newcommand{\fa}{\:\:\forall\:}
\newcommand{\p}[3]{#1_{#2_{#3}}}  
\newcommand{\norm}[3]{\lVert #1 \rVert_{#2}^{#3}}
\newcommand{\D}[1]{\mathcal{D}^{#1}}
\newcommand{\mbf}[1]{\mathbf{#1}}
\newcommand{\mbb}[1]{\mathbb{#1}}
\newcommand{\bmat}[1]{\begin{bmatrix}#1\end{bmatrix}}
\newcommand{\inner}[2]{\left\langle #1, #2 \right\rangle}
\newcommand{\bs}[1]{\boldsymbol{#1}} 
\newcommand{\tnf}[1]{\textnormal{#1}}

\newcommand{\T}{\mathcal{T}}
\newcommand{\R}{\mathcal{R}}
\newcommand{\U}{\mathcal{U}}
\newcommand{\Z}{\mathcal{Z}}

\newcommand{\Ncal}{\mathcal{N}}
\newcommand{\Pop}{\mathcal{P}}
\newcommand{\Q}{\mathcal{Q}}
\newcommand{\Sm}{\mathcal{S}}
\newcommand{\C}{\mathcal{C}}
\newcommand{\M}{\mathcal{M}}
\newcommand{\K}{\mathcal{K}}
\newcommand{\Hop}{\mathcal{H}}

\newcommand{\G}{\mathcal{G}}

\newtheorem{definition}{Definition}
\newtheorem{lemma}{Lemma}

\newtheorem{theorem}{Theorem}

\renewcommand{\baselinestretch}{0.98}

\newcounter{MYtempeqncnt}

\title{\LARGE \bf Distributed Sum-of-Squares Programming for Local Stability Analysis of Polynomial PDEs} 

\author{Carl R. Richardson  \hspace{3mm} Declan S. Jagt \hspace{3mm} Matthew M. Peet \hspace{3mm} Antonis Papachristodoulou%
\thanks{This work was supported in part by the National Science Foundation under Grant NSF CCF-2429973 and in part by the EPSRC under project UKRI2108. For the purpose of Open Access, the author has applied a CC BY public copyright licence to any Author Accepted Manuscript (AAM) version arising from this submission.}%
\thanks{DS. Jagt and MM. Peet are with the Department of Mechanical and Aerospace Engineering, Arizona State University, Tempe, AZ 85281, USA}%
\thanks{CR. Richardson and A. Papachristodoulou are with the Department of Engineering Science, University of Oxford, Oxford, OX1 3PJ, UK}%
}
   
\begin{document}

\maketitle
\thispagestyle{empty}
\pagestyle{empty}

\begin{abstract}
It has recently been shown that the evolution of a state, described by a Partial Differential Equation (PDE), can be more conveniently represented as the evolution of the state's highest spatial derivative (the ``fundamental state''), which lies in $L_2$ and has no boundary conditions or continuity constraints. For linear PDEs, this yields a Partial Integral Equation (PIE) parametrized by Partial Integral (PI) operators mapping the fundamental state to the PDE state. In this paper, we show that for polynomial PDEs, the dynamics of the fundamental state can be compactly expressed as a distributed polynomial in the fundamental state, parametrized by a new tensor algebra of PI operators acting on the tensor product of the fundamental state. We further define a sum-of-squares (SOS) parameterization of the distributed polynomial and use this to construct a distributed SOS program, for testing local stability of polynomial PDEs.
\end{abstract}

\section{Introduction}
In this paper, we address the problem of representation and local stability analysis of polynomial Partial Differential Equations (PDEs). Such PDEs are commonly used to model physical processes including quantum mechanics, condensed matter, fluid dynamics, population growth, and wave propagation \cite{logan2004applied}. However, some features of the standard PDE representation create challenges for analysis and simulation. For example, consider Burgers' Equation with the additional reaction term and Dirichlet Boundary Conditions (BCs):
\begin{equation} \label{eq:modBurgs}
    \begin{split}
        \p{u}{t}{}(t,s) &= \p{u}{ss}{}(t,s)- u(t,s) \p{u}{s}{}(t,s) + ru(t,s) \\
        u(t,0) &= u(t,1) = 0 \hspace{3mm} t \geq 0\, \hspace{1mm} s\in[0,1]
    \end{split}
\end{equation}
Using integration by parts, the Poincar\'e inequality and invoking the BCs, the Lyapunov Functional (LF) $V(u) = \norm{u}{L_{2}}{2} \geq 0$ can verify global stability of the zero solution if $r < \pi^2$ \cite{valmorbida2014semi}. This approach works well for \eqref{eq:modBurgs} since $\dot{V}(u)$ only depends on $u, u_{s}$, allowing the Poincar\'e inequality to be applied. More generally, this is not applicable and the representation of PDEs as polynomial functions of $u, u_{s}, \dots, \p{u}{s}{n}$ complicates the stability analysis, even for simple LF candidates.

Now suppose a Dirichlet BC is replaced by a Neumann BC, $\p{u}{s}{}(t,1)=0$. Despite the PDE and LF candidate remaining the same, the derivative of this candidate along the solutions to \eqref{eq:modBurgs} is different, meaning the previously obtained stability test is no longer valid. Furthermore, it is unclear from this representation how the BCs impact stability properties and how they can be accounted for when testing the fitness of a LF candidate.

Because of these difficulties, most prior work focused on the stability analysis of specific nonlinear PDEs and BCs. For example, stability conditions have been derived for the Navier-Stokes  \cite{goulart2012global, huang2015sum, fuentes2022global} and the Kuramoto-Sivashinsky equations with periodic BCs \cite{goluskin2019bounds}. These results offer some, but limited insight into how to test stability of other nonlinear PDEs. Other work has studied limited classes of PDEs including the wave equation with a bounded nonlinear term \cite{fridman2016new, prieur2026stability} and a class of 2nd order parabolic PDEs \cite{valmorbida2015stability, mironchenko2019monotonicity}. These results focus on specific applications, and the proposed stability conditions may be challenging to enforce in practice. Finally, moment methods have been applied to test bounds on energy functionals of nonlinear PDEs \cite{korda2022moments, marx2018moment}, but an intuitive algorithm to test stability is missing.

To address the challenges posed by the representation of nonlinear PDEs, a new approach represents the dynamics of the PDE state as a Partial Integral Equation (PIE)~\cite{peet2021partial}. This representation expresses the dynamics of the PDE state only in terms of its highest spatial derivative, known as the ``fundamental state''. Unlike the PDE state, which lies in a Sobolev space, the fundamental state lies in $L_{2}$ and is free of BCs and continuity constraints. For linear PDEs, the PIE is parametrized by PI operators which map the fundamental state to the PDE state and its derivatives. Through these operators, the BCs and continuity constraints are embedded in the PIE, rather than constraining the PDE state. 

More recently, the PIE framework was extended to develop a global stability test for quadratic PDEs \cite{jagt2023pie}. In this setting, the PIE was parametrized by a special case of the more general ``Tensor-PI'' operators presented in this work. Leveraging the PIE representation, a suite of SDP tools have been developed for stability analysis, estimation, and control of PDEs. A subset of examples are referenced here \cite{jagt20222, wu2023h, braghini2025static}, all of which are implemented in the open-source software package PIETOOLS \cite{shivakumar2025pietools}.

\textbf{Contribution}: Inspired by the SOS framework for polynomial ODEs \cite{prajna2002introducing}, this article defines a PI representation of distributed polynomials in the fundamental state. This representation enables parameterization of the PIE, such that it can represent 1D polynomial PDEs, and defines distributed semi-algebraic sets, like the $L_{2}$-ball. A SOS parameterization of the distributed polynomial is also defined, enabling positivity of distributed polynomials to be enforced through equality constraints. Leveraging such distributed polynomials, a distributed SOS program is constructed to test local stability on the $L_{2}$-ball. This test is suitable for all 1D polynomial PDEs with linear BCs.

\section{Notation} \label{sec:not}
For $\Omega := [a,b] \subset \Real$, let $L_{2}^{m \times n}[\Omega]$ and $L_{\infty}^{m \times n}[\Omega]$ denote the spaces of square-integrable and bounded $\Real^{m \times n}$-valued functions on $\Omega$. For $\alpha, k \in \N$, let $\p{x}{s}{\alpha} := \frac{\partial^{\alpha}}{\partial s^{\alpha}}x$, $\D{k}x(s) := [x(s)\ \p{x}{s}{}(s)\ \cdots\ \p{x}{s}{k}(s)]'$, and denote the degree $k$ Sobolev space on $\Omega$ by $H_{2}^{k}[\Omega]$.
For $u \in \Real^{m}$, $v \in \Real^{n}$, let $u \otimes v \in \Real^{mn}$ denote the Kronecker product. For $x \in L_{2}^{m}[a,b]$, $y \in L_{2}^{n}[c,d]$, define the tensor product $x \otimes y \in L_{2}^{mn}[[a,b] \times [c,d]]$ as $[x \otimes y](s_{1},s_{2})= x(s_{1}) \otimes y(s_{2})$. For $x \in L_{2}[\Omega]$, a degree-$d$ distributed monomial is defined as $d$ tensor products $x^{\otimes d} := x \otimes \cdots \otimes x \in L_{2}[\Omega^{d}]$ and the degree-$d$ distributed monomial basis on $x$ is denoted by $\Z_{d}(x) := [1 \ x \ x^{\otimes 2} \ \cdots \ x^{\otimes d}]'$, where the degree-$0$ term is omitted when convenient. Let $\Real[\bs{\theta}_{k}]$ denote the polynomial ring in $k$ variables. For $G \in \Real^{n \times m}$, row $i$ is denoted by $G_{i,:}$. Let $\Sm_{w} := \{ (i,j) \in \N^{2} : \hspace{1mm} i + j = w \}$ and denote a subset of all $n!$ permutations by $A_{n} \subseteq \{1,\dots,n\}^{n}$. For $\alpha \in A_{n}$ and $k, j \in \{1, \dots, n\}$, the indicator functions used are defined below, where $\mbb{I}_{1}(\theta_{1}) := 1$ if $n=1$.
\begin{align*}
    \mbf{1}(x) &:= \begin{cases} 1 & \text{if} \ x=0 \\ 0 & \text{else} \end{cases} & 
    \mbf{I}_{\epsilon}(x) &:= \begin{cases} 1 & \text{if} \ x \geq \epsilon \\ 0 & \text{else} \end{cases} \\
    \gamma(\alpha_{k},j) &:= \begin{cases} 1 & \text{if}\ k\leq j \\ 2 & \text{if}\ k > j \end{cases} & \mbb{I}_{\alpha}(\bs{\theta}_{n}) &:= \prod_{j=1}^{n-1} \hspace{-1mm} \mbf{I}_{0}(\theta_{\alpha_{j+1}} \hspace{-1mm} - \theta_{\alpha_{j}})
\end{align*}
Finally, nested integrals are compactly denoted by
\begin{equation*}
    \int_{\Omega^{k}} K(\bs{\theta}_{k}) d \bs{\theta}_{k} = \int_{a}^{b} \cdots \int_{a}^{b} K(\theta_{1}, \dots, \theta_{k}) d \theta_{k} \dots d \theta_{1}
\end{equation*}

\section{A Map From Fundamental State to PDE State} \label{sec:funstate}
For a domain $\Omega := [a,b]$, we consider the $n$th order, 1D, polynomial PDE 
\begin{equation} \label{eq:pde}
    \p{u}{t}{}(t,s) = c(s)^{\top} \U(u, \p{u}{s}{}, \dots, \p{u}{s}{n}) \quad s \in \Omega
\end{equation}
where $u(t) \in X_{B}[\Omega]$ is the PDE state, $c \in L_{\infty}^{m}[\Omega]$, and $\U$ is a vector of $m \in \N$ monomials in $u, \p{u}{s}{}, \dots, \p{u}{s}{n}$
\begin{equation} \label{eq:pde_mon}
    \hspace{-2mm} \U_{i}(u, \p{u}{s}{}, \dots, \p{u}{s}{n}) := u(t,s)^{\alpha_{i0}} \p{u}{s}{}(t,s)^{\alpha_{i1}} \hspace{-4mm} \dots \p{u}{s}{n}(t,s)^{\alpha_{in}}
\end{equation}
defined by the coefficient matrix $A := [\alpha_{ij}] \in \N^{m \times (n+1)}$. The PDE is assumed to have linear BCs; hence $X_{B}[\Omega]$ is
\begin{equation} \label{eq:X_B}
    \hspace{-2mm} X_{B}[\Omega] \hspace{-1mm} := \hspace{-1mm} \bigg\{ u \in H_{2}^{n}[\Omega] \hspace{-1mm}: B \hspace{-1mm} \begin{bmatrix} \D{n-1}u(a) \\ \D{n-1}u(b) \end{bmatrix} \hspace{-1mm} = 0, B \in \Real^{n \times 2n} \hspace{-1mm} \bigg\}
\end{equation}
Dirichlet, Neumann, and Robin BCs can be enforced by setting $B$ appropriately. In summary, a polynomial PDE is parametrized by $\{c, B, A \}$ from \eqref{eq:pde}-\eqref{eq:X_B}.

\begin{definition}[Classical solution to the PDE] \label{def:pde_sol}
    For a given $u_{0} \in X_{B}[\Omega]$, $u$ is a classical solution to the polynomial PDE defined by $\{c, B, A\}$ if $u$ is Fr\'echet differentiable, $u(0)=u_{0}$, and \eqref{eq:pde} is satisfied by $u(t) \in X_{B}[\Omega] \fa t \geq 0$.
\end{definition}

In this section, we assume full-rank BCs (see Lem.~\ref{lem:fr}) and show how a PDE of the form \eqref{eq:pde} can be defined in terms of PI operators acting on the fundamental state. The advantage of this representation is that the BCs and continuity constraints from \eqref{eq:X_B} are moved into the representation of the dynamics. We now recap the definition of a PI operator \cite{peet2021partial}.

\begin{definition} \label{def:PI}
    For a given domain $\Omega := [a,b]$ and $k,r \in \N$, the parameter space $\Ncal_{3}^{k \times r}[\Omega] := L_{\infty}^{k \times r}[\Omega] \times L_{2}^{k \times r}[\Omega^{2}] \times L_{2}^{k \times r}[\Omega^{2}]$ is defined. For $x \in L_{2}^{r}[\Omega]$ and any $\{R_{0}, R_{1}, R_{2} \} \in \Ncal_{3}^{k \times r}$, the associated PI operator is defined as
    \begin{equation*} 
        (\R x)(s) \hspace{-0.5mm} = \hspace{-0.5mm} R_{0}(s) x(s) + \hspace{-1mm}  \int_{a}^{s} \hspace{-3mm}  R_{1}(s, \theta) x(\theta) d \theta + \hspace{-1mm}  \int_{s}^{b} \hspace{-3mm} R_{2}(s, \theta) x(\theta) d \theta
    \end{equation*}   
    where $\R : L_{2}^{r}[\Omega] \to L_{2}^{k}[\Omega]$.
\end{definition}

For $\{R_{0}, R_{1}, R_{2}\} \in \Ncal_{3}^{k \times r}$, define $\R := \Pop_{\{R_{0}, R_{1}, R_{2}\}} \in \Pi_{3}^{k \times r}$. If $\R = \Pop_{\{0, R_{1}, R_{2}\}}$, then $\R \in \Pi_{2}^{k \times r}$. We refer to these as 3-PI and 2-PI operators, respectively. Both operators form *-algebras \cite{jagt2025state}; hence, the summation, composition, and adjoint can be computed directly from their parameters.

Now consider solutions to \eqref{eq:pde}. Since $u(t) \in X_{B} \subset H_{2}^{n} \fa t \geq 0$, the $n$th order spatial derivative exists and need only satisfy $\p{u}{s}{n}(t) \in L_{2}$. Hence, it is free of BCs and continuity constraints. We refer to this derivative $v(t) := \p{u}{s}{n}(t)$ as the \emph{fundamental state}. A map from the PDE state space, $X_{B}$, to the fundamental state space, $L_{2}$, is defined by the differential operator $\frac{\partial^{n}}{\partial_{s}^{n}}: X_{B} \to L_{2}$. Assuming full-rank BCs, an inverse map exists and is defined by a 2-PI operator $\T: L_{2} \to X_{B}$. We recall the following Lemma from \cite{shivakumar2024extension} which provides explicit formulae mapping $B$ to the operators $\T$, as well as $\R_{j}:=\partial_{s}^{j} \circ \T$.

\begin{lemma} \label{lem:fr}
    Let $\Omega := [a,b]$ and $X_B$ be as in \eqref{eq:X_B}, where $B_1, B_2 \in \Real^{n \times n}$ such that $B = [B_1 \hspace{1mm} B_2]$. For $s \in \Omega$ and $i,j \in \{0,\dots,n-1\}$, define
    \begin{align}
      \hspace{-1mm} G_{ij}(s) &:= \begin{cases}
                        \frac{(s-a)^{j-i}}{(j-i)!} & j \geq i \\
                        0 & else
                    \end{cases}
        &
        h_{i}(s) &:= \frac{s^{n-i-1}}{(n-i-1)!}
    \end{align}
    and set $G(s) := [G_{ij}(s)]$, $H(s) := [h_0(s)\ \dots\ h_{n-1}(s)]'$.

    If $B_1 + B_2 G(b)$ is full rank, then the inverse operator $\mathcal{T}: L_2[\Omega] \to X_B[\Omega]$ exists with $\mathcal{T} = \mathcal{P}_{\{T_1,T_2\}} \in \Pi_2$, where
    \begin{equation} \label{eq:T}
        \begin{split}
             T_{1}(s, \theta) &:= h_{0}(s-\theta) - G_{0, :}(s) F(\theta) \\
             T_{2}(s, \theta) &:= -G_{0, :}(s) F(\theta)   
        \end{split}
    \end{equation}
    and $F(\theta) := (B_1 + B_2 G(b))^{-1} B_2 H(b-\theta)$. Defining $\mathcal{R}_j := \frac{\partial^j}{\partial s^j} \circ \mathcal{T} \in \Pi_2$, we have for all $u \in X_B[\Omega]$ and $v \in L_2[\Omega]$
    \begin{align} \label{eq:derivs}
    \hspace{-1mm} \p{u}{s}{j}(s) &= (\R_{j} [\frac{\partial^{n}}{\partial s^{n}} u])(s) &
    (\R_{j}v)(s) &= (\frac{\partial^{j}}{\partial s^{j}} [\T v])(s)
    \end{align}
    \end{lemma}

\begin{proof}
    Refer to Thm. 10 and Thm. 12 in \cite{shivakumar2024extension}.
\end{proof}

Using this Lemma, we can express the PDE state and its derivatives in terms of the fundamental state. Clearly, an identity operator $\R_{n} = \Pop_{ \{1, 0, 0 \} } \in \Pi_{3}$ can also be defined such that $\p{u}{s}{n}(s) = (\R_{n}v)(s)$. Substituting these relations into \eqref{eq:pde}, we find that if $u$ satisfies the PDE, then $v$ satisfies
\begin{equation} \label{eq:pie1}
    \p{(\T v)}{t}{}(t,s) = c(s)^{\top} \U \left(\T v, \p{(\T v)}{s}{}, \dots, \p{(\T v)}{s}{n} \right)
\end{equation}
where $v(t) \in L_{2}[\Omega] \fa t \geq 0$. This PIE representation is free of BCs and continuity constraints, and is only represented in terms of the fundamental state. In the next section, we define the representation of a distributed polynomial and show how \eqref{eq:pie1} can be expressed in that format.

\section{Polynomials on a Distributed State} \label{sec:poly}
In this section, we aim to generalize \cite[Section IV]{jagt2023pie} to enable the representation of higher order tensor products of PI operators, such as $(\T v)^{\otimes \alpha_{i0}} (\R_{1} v)^{\otimes \alpha_{i1}} \dots (\R_{n} v)^{\otimes \alpha_{in}}$ from \eqref{eq:pie1}, and integrals over these higher order tensor products, which arise in the representation of LFs and distributed semi-algebraic sets. We first define the distributed polynomial in its most general form and then define two classes of operators to parametrize the distributed polynomial such that it can represent both terms. The key advantage of the distributed polynomial representation is that it is unique, which will later enable us to enforce equality constraints.

\begin{definition} \label{def:dispoly}
    For a given domain $\Omega := [a,b]$, we define a DP of degree $d \in \N$ on $x \in L_{2}[\Omega]$ as
    \begin{equation*}
        p(x) := \C \Z_{d}(x) := \sum_{i=0}^d \C_i x^{\otimes i}
    \end{equation*}
    where $\Z_{d}(x)$ is the distributed monomial basis on $x$ and $\C:= [\C_{0} \hspace{1mm} \cdots \hspace{1mm} \C_{d}]$ with $\C_{i}$ being an operator acting on the degree $i$ distributed monomial.
\end{definition}

In this paper, $\C_0 \in \Real$ and $\{\C_1, \dots \C_d\}$ are parametrized by either Tensor-PI or Functional-PI operators, which are presented next.

\subsection{Tensor-PI Operators} \label{sec:TPI}
The class of Tensor-PI operators is used to represent tensor products of 3-PI operators, mapping elements of $L_{2}[\Omega^{d}]$ to elements of $L_{2}^{k}[\Omega]$.
\begin{definition} \label{def:TPI}
     For a given domain $\Omega:= [a,b]$ and $d, k \in \N$, the operator $\Hop: L_{2}[\Omega^{d}] \to L_{2}^{k}[\Omega]$ is a Tensor-PI operator if there exist $m, k_{ij} \in \N$ and 3-PI operators $\R_{ij}: L_{2}[\Omega] \to L_{2}^{k_{ij}}[\Omega]$ such that
    \begin{equation*} 
        (\Hop[x_{1} \otimes \cdots \otimes x_{d}])(s)
        = \sum_{i=1}^{m} (\R_{i1}x_{1})(s) \otimes \cdots \otimes (\R_{id} x_{d})(s)
    \end{equation*}
    for all $x_{j} \in L_{2}[\Omega]$, where $k = \prod_{j=1}^{d} k_{ij} \fa i$.
\end{definition}
The operator $\Hop$ in Def. \ref{def:TPI} is denoted by $\Hop = \sum_{i=1}^{m} \Hop_{i}$ where $\Hop_{i} = \R_{i1} \otimes \cdots \otimes \R_{id}$\footnote{Formally, this notation is more akin to the face-splitting product.} or $\Hop_{i} = \R_{i1}^{\otimes d}$ if the operators $\R_{i1}, \dots, \R_{id}$ are the same. This includes all operators that can equally be represented as a linear combination of: (i) operators that act on a tensor product; (ii) tensor products of operators. If an operator $\Hop$ satisfies Def. \ref{def:TPI}, then $\Hop$ is said to be a T-PI operator where $\Hop \in \Pi_{T}^{k}$. This representation enables tensor products of PI operators to be represented as distributed polynomials.

\subsection{Functional-PI Operators}
Functional operators map function-valued objects to real scalar values. Although such functionals $\K: L_{2}[\Omega^d] \to \Real$ can always be expressed as an integral over any bounded kernel, $\K w := \int_{\Omega} K(\theta) w(\theta) d \theta$, for our purposes, the kernel is restricted to being polynomial. This class of operators is used to parametrize the distributed polynomial such that it can represent LFs and distributed semi-algebraic sets.

\begin{definition} \label{def:FPI}
    For a given domain $\Omega :=[a,b]$ and $d \in \N$, the operator $\K: L_{2}[\Omega^{d}] \to \Real$ is a Functional-PI operator if there exist $K_{\alpha} \in \Real[\bs{\theta}_{d}]$ for $\alpha \in A_{d}$ such that
    \begin{equation}
        \hspace{-1mm} \K x = \sum_{\alpha \in A_{d}} \hspace{-1mm} \int_{\Omega^{d}} \mathbb{I}_{\alpha}(\bs{\theta}_{d}) K_{\alpha}(\bs{\theta}_{d}) x(\bs{\theta}_{d}) d \bs{\theta}_{d} \hspace{2mm} \fa x \in L_{2}[\Omega^{d}]
    \end{equation}
\end{definition}    

Note that $A_{d}$ and $\mathbb{I}_{\alpha}(\bs{\theta}_{d})$ are defined in Section \ref{sec:not}. If an operator $\K$ satisfies Def. \ref{def:FPI}, then $\K$ is said to be an F-PI operator where $\K \in \Pi_{F}$. Integrals over T-PI operators are not obviously in the form of an F-PI operator; however, Lem.~\ref{lem:int_split} shows they are.

\begin{lemma} \label{lem:int_split}
    For a given domain $\Omega :=[a,b]$ and $d \in \N$, then any integrable function $K:\Omega^{d} \to \Real$ satisfies
    \begin{equation}
        \int_{\Omega^d} K(\bs{\theta}_{d}) \, d \bs{\theta}_{d}
        = \sum_{\alpha\in A_{d}} \int_{\Omega^{d}} \mathbb{I}_{\alpha}(\bs{\theta}_{d}) K(\bs{\theta}_{d}) \, d \bs{\theta}_{d}
    \end{equation}
\end{lemma}

\begin{proof}
    Refer to Appendix.
\end{proof}

Since integral operators are linear, it is not difficult to see that the class of F-PI operators is closed under linear combinations. In particular, if $\K, \G :L_{2}[\Omega^{d}] \to \Real$ are F-PI operators defined by parameters $\{K_{\alpha}\}$ and $\{G_{\alpha}\}$, respectively, then $\lambda \K + \mu \G \in \Pi_{F}$ for any $\lambda, \mu \in \Real$, defined by $\{\lambda K_{\alpha} + \mu G_{\alpha}\}$. In addition, a tensor product of F-PI operators can be defined as follows, acting on tensor products of distributed monomials.

\begin{definition}
    For a given domain $\Omega :=[a,b]$ and $d,r \in \N$, let $\K:L_{2}[\Omega^{d}]\to\Real$ and $\G:L_{2}[\Omega^{r}]\to\Real$ be F-PI operators, respectively defined by parameters $\{K_{\alpha}\}_{\alpha\in A_{d}}$ and $\{G_{\alpha}\}_{\alpha\in A_{r}}$. We define $\K\otimes\G:L_{2}[\Omega^{d+r}]\to\Real$ by
    \begin{align*}
        (\K\otimes G)x
        = \hspace{-1mm} \sum_{\alpha\in A_{d+r}} \hspace{-1mm} \int_{\Omega^{d+r}} \hspace{-5mm} \mbb{I}_{\alpha}(\bs{\theta}_{d+r}) H_{\sigma(\alpha)}(\bs{\theta}_{d+r}) x(\bs{\theta}_{d+r}) d\bs{\theta}_{d+r}
    \end{align*}
    where $H_{\sigma(\alpha)}(\bs{\theta}_{d}, \bs{\theta}_{r}) = K_{\sigma_{1}(\alpha)}(\bs{\theta}_{d})G_{\sigma_{2}(\alpha)}(\bs{\theta}_{r})$ for each $\alpha\in A_{d}$, with $\sigma_{1}(\alpha):=\alpha\setminus\{d+1,\hdots,d+r\}\in A_{d}$ and $\sigma_{2}(\alpha):=\alpha\setminus\{1,\hdots,d\}-d\in A_{r}$.
\end{definition}

Given this tensor product of F-PI operators, the following Lemma shows that the product of distributed polynomials, parametrized by such F-PI operators, can be expressed in terms of a tensor product of the F-PI operators, acting on a tensor product of the distributed monomials.
\begin{lemma}\label{lem:polynomialProduct}
    For given $\Omega :=[a,b]$, $d \in \N$, define $\K:=[\K_{1} \hspace{2mm} \cdots \hspace{2mm} \K_{d}]$ and $\G:=[\G_{1} \hspace{2mm} \cdots \hspace{2mm} \G_{d}]$ by F-PI operators $\K_{i},\G_{i}:L_{2}[\Omega^{i}]\to\Real$. Let $\C:=[0 \hspace{2mm} \C_{2} \hspace{2mm} \cdots \hspace{2mm} \C_{2d}]$ where $\C_{k}:=\sum_{(i,j)\in \Sm_{k}}\K_{i}\otimes \G_{j}$ for $k\in\{2,\dots,2d\}$.
    If $p(x)=\K\Z_{d}(x)$ and $q(x)=\G\Z_{d}(x)$, then $p(x)q(x)=\C\Z_{2d}(x) \fa x \in L_{2}[\Omega]$.
\end{lemma}
\begin{proof}
    Refer to Appendix.
\end{proof}

By Lemma~\ref{lem:polynomialProduct}, the class of distributed polynomials parametrized by F-PI operators is closed under the multiplication operation, thus forming a ring (like polynomials in finite variables). We now define the subset of F-PI operators with the particular form which arises from integrals over T-PI operators. This is referred to as the ``vectorization" of scalar T-PI operators.

\begin{definition} \label{def:vec}
   For $\Omega := [a,b]$, $i \in \{1,\dots,d\}$ and a given set of polynomial parameters $R_{i,1}, R_{i,2} \in \Real[s,\theta_{i}]$, define the associated functions $H_{\alpha} \in \Real[\bs{\theta}_{d}]$ for each $\alpha \in A_{d}$ as
    \begin{align*}
    	H_{\alpha}(\bs{\theta}_{d})
        &:= \hspace{-1mm} \sum_{j=0}^{d} \int_{\theta_{\alpha_{j}}}^{\theta_{\alpha_{j+1}}} \hspace{-2mm} R_{1,\gamma(1,j)}(s,\theta_{1}) \cdots R_{d,\gamma(d,j)}(s,\theta_{d}) \, ds
    \end{align*}
    where $\theta_{\alpha_{0}}:=a$, $\theta_{\alpha_{d+1}}:=b$. For a given tensor product $\Hop := \R_{1} \otimes \cdots \otimes \R_{d}: L_{2}[\Omega^{d}]\to L_{2}[\Omega]$, where each $\R_{i} = \Pop_{ \{R_{i,1}, R_{i,2} \} } \in \Pi_{2}$, define $\tnf{vec}(\Hop): L_{2}[\Omega^{d}] \to \Real$ as
    \begin{align*}
        \tnf{vec}(\Hop) x := \sum_{\alpha\in A_{d}} \int_{\Omega^{d}} \mathbb{I}_{\alpha}(\bs{\theta}_{d}) H_{\alpha}(\bs{\theta}_{d}) x(\bs{\theta}_{d}) \, d \bs{\theta}_{d}
    \end{align*}
    Then, for a given scalar T-PI operator $\Hop:= \sum_{i=1}^{m} \Hop_{i}$ where $\Hop_{i} = \R_{i,1} \otimes \cdots \otimes \R_{i,d}: L_{2}[\Omega^{d}] \to L_{2}[\Omega]$, we define $\tnf{vec}(\Hop):= \sum_{j=1}^{m} \tnf{vec}(\Hop_{i})$. 
\end{definition}

Note that the indicator functions $\gamma(\alpha_{k},j)$, $\mathbb{I}_{\alpha}(\bs{\theta}_{d})$ are defined in Section \ref{sec:not}. Finally, the following Lemma proves that the vectorization of a scalar T-PI operator corresponds to the integral over the scalar T-PI operator, enabling integrals over scalar T-PI operators to be represented as distributed polynomials.

\begin{lemma} \label{lem:vec}
    For a given domain $\Omega := [a,b]$, any scalar T-PI operator $\Hop: L_{2}[\Omega^{d}] \to L_{2}[\Omega]$ satisfies
    \begin{equation*}
        \int_{\Omega} (\Hop x)(s) \, ds = \tnf{vec}(\Hop) x \quad \fa x \in L_{2}[\Omega^{d}]
    \end{equation*}
\end{lemma}

\begin{proof}
    Refer to Appendix.
\end{proof}

Linear combinations and multiplication of distributed polynomials, parametrized by F-PI operators, arise when constructing the distributed SOS program for local stability analysis of polynomial PDEs. The closure and vectorization results presented in this section enable the equality constraints to be enforced. 

\subsection{Example: Distributed Semi-algebraic Sets}
In the upcoming sections, we show how T-PI and F-PI operators can be used to parametrize distributed polynomials for the representation of polynomial PDEs and LFs. For now, we present an example of how these operators can be used to represent distributed semi-algebraic sets, such as the $L_{2}$-ball of radius $r$, as distributed polynomials.

For a domain $\Omega := [a,b]$ and $r>0$, define the ball of radius $r$ in the PDE domain
\begin{equation} \label{eq:Xbr}
    X_{B,r}[\Omega] := \left\{ u\in X_{B}[\Omega]: \hspace{1mm} r^{2} - \norm{u}{L_{2}}{2} \geq 0 \right\}
\end{equation}
Using the inverse map $\T: L_{2}[\Omega] \to X_{B}[\Omega] \in \Pi_{2}$, applying Def. \ref{def:TPI}, and Lem. \ref{lem:vec}, the constraint is equal to
\begin{align*}
    r^{2} - \norm{\T v}{L_{2}}{2} &= r^{2} - \int_{\Omega} (\T v)(s) (\T v)(s) ds \\
    &= r^{2} - \int_{\Omega} \left((\T \otimes \T)v^{\otimes 2}\right) \hspace{-1mm} (s) ds \\
    &= r^{2} - \underbrace{\tnf{vec} \left( \T \otimes \T \right)}_{\in \Pi_{F}} v^{\otimes 2} = \G \Z_{2}(v)
\end{align*}
a distributed polynomial in the fundamental state, where $\G = [r^{2} \hspace{2mm} 0 \hspace{2mm} -\tnf{vec}( \T \otimes \T)]$. As a result, the set \eqref{eq:Xbr} can be represented by the following distributed semi-algebraic set on the fundamental state
\begin{equation} \label{eq:L2r}
    L_{2,r}[\Omega] := \left\{ v \in L_{2}[\Omega] : \hspace{1mm} \G \Z_{2}(v) \geq 0 \right\}
\end{equation}
\vspace{-5mm}
%
%
\begin{figure*}[!t]
\setcounter{MYtempeqncnt}{\value{equation}} 
\setcounter{equation}{15}
\begin{equation} \label{eq:Qop}
    \Q := \hspace{-1mm} \begin{bmatrix}
          \Q_{11} & \Q_{12} & \cdots & \Q_{1d} \\
          \Q_{21} & \Q_{22} & \cdots & \Q_{2d} \\
          \vdots & \vdots & \ddots & \vdots \\
          \Q_{d1} & \Q_{d2} & \cdots & \Q_{dd}
      \end{bmatrix}  
      \hspace{-1mm} := \hspace{-1mm}
      \begin{bmatrix}
          \hat{\U} &  &  &  & \\
             & \hat{\U}^{\otimes 2} &  &  & \\
             &  &  \ddots &  & \\
             & & & & \hat{\U}^{\otimes d}
      \end{bmatrix}^{*}
      \hspace{-1mm}
      \circ
      \underbrace{\begin{bmatrix}
          \bar{Q}_{11} & \bar{Q}_{12} & \cdots & \bar{Q}_{1d} \\
          \bar{Q}_{21} & \bar{Q}_{22} & \cdots & \bar{Q}_{2d} \\
          \vdots & \vdots & \ddots & \vdots \\
          \bar{Q}_{d1} & \bar{Q}_{d2} & \cdots & \bar{Q}_{dd}
      \end{bmatrix}}_{:= \bar{Q}}
      \circ
      \underbrace{\begin{bmatrix}
          \hat{\U} &  &  &  & \\
             & \hat{\U}^{\otimes 2} &  &  & \\
             &  &  \ddots &  & \\
             & & & & \hat{\U}^{\otimes d}
      \end{bmatrix}}_{:= \hat{U}_{d}}
      \vspace{-2mm}
\end{equation}
\setcounter{equation}{\value{MYtempeqncnt}} 
\hrulefill 
\vspace{-5mm}
\end{figure*}
%
%
\section{A PIE Representation of Polynomial PDEs} \label{sec:PIEs}
We now revisit the PIE representation of the evolution of $v$ in \eqref{eq:pie1}. Using the T-PI operators introduced in the previous section, this PIE can be expressed as a distributed polynomial. In particular, we show that the evolution of $v$ is governed by a distributed polynomial of the form
\begin{equation} \label{eq:pie}
    \T \p{v}{t}{}(t,s) = \C \Z_{d}(v) \quad s \in \Omega
\end{equation}
where $v(t) \in L_{2}[\Omega]$, $\T \in \Pi_{2}$, and $\C := [\C_{1} \hspace{1mm} \dots \hspace{1mm} \C_{d}]$ is parametrized by T-PI operators $\C_{k}: L_{2}[\Omega^{k}] \to L_{2}[\Omega] \in \Pi_{T}$. With reference to the polynomial PDE defined by \eqref{eq:pde}-\eqref{eq:X_B}, the operators $\C_{k}$ are given by
\begin{equation} \label{eq:ck}
    \C_{k} := \sum_{i=1}^{m}  \mathbf{1}(k - \alpha_{i}) \M_{c_{i}} \circ \Hop_{i}
\end{equation}
where $\alpha_{i} := \sum_{j=0}^{n} \alpha_{ij}$ denotes the sum over row $i$ of matrix $A$; $d = \max \{\alpha_{1}, \dots, \alpha_{m} \}$; $\M_{c_{i}} := \Pop_{ \{c_{i},0,0\} }$ is the multiplier operator associated to $c_{i} \in L_{\infty}[\Omega]$; and $\Hop_{i} := \T^{\otimes \alpha_{i0}} \otimes \R_{1} ^{\otimes \alpha_{i1}} \otimes \dots \otimes \R_{n}^{\otimes \alpha_{in}}$ is the T-PI operator associated to the monomial $\U_{i}$, also defined by $A$. We note here $\M_{c_{i}} \circ \Hop_{i} \in \Pi_{T}$ follows from the fact that, by definition, $\Pi_{T}$ is closed under linear combinations.

\begin{definition}[Classical solution to the PIE] \label{def:pie_sol}
    For a given $v_{0} \in L_{2}[\Omega]$, $v$ is a classical solution to the polynomial PIE defined by $\{\T, \C \}$ if $v$ is Fr\'echet differentiable, $v(0)=v_{0}$, and \eqref{eq:pie} is satisfied by $v(t) \in L_{2}[\Omega] \fa t \geq 0$.
\end{definition}

The following Lemma proves there exists an invertible map between classical solutions to the polynomial PDE \eqref{eq:pde} and classical solutions to the associated PIE \eqref{eq:pie}.

\begin{lemma}
    Suppose that $B \in \Real^{n \times 2n}$ satisfies Lemma \ref{lem:fr} and for $j \in \{0, \dots, n\}$ the operators $\T, \R_{j} \in \Pi_{3}$ are defined accordingly. Let $c \in L_{\infty}^{m}[\Omega]$, $A \in \N^{m \times (n+1)}$ and define the operators $\C_{k}$ as in \eqref{eq:ck}. Then, $v$ is a classical solution to the PIE, defined by $\{ \T, \C\}$, with initial state $v_{0}$ if and only if $\T v$ is a classical solution to the PDE, defined by $\{ c, B, A\}$, with initial state $\T v_{0}$. Conversely, $u$ is a classical solution to the polynomial PDE, defined by $\{c,B,A\}$, with initial state $u_{0}$ if and only if $\p{u}{s}{n}$ is a classical solution to the polynomial PIE, defined by $\{\T,\C\}$, with initial state $(u_{0})\p{}{s}{n}$.
\end{lemma}

\begin{proof} By Lem. \ref{lem:fr}, the operators $\T, \R_{j} \in \Pi_{3}$ for $j \in \{0, 1, \dots, n\}$ exist and are defined, along with operators $\C_{k} \in \Pi_{T}$ for $k \in \{1, \dots, d\}$, as defined in \eqref{eq:ck}. Exploiting these definitions and that of a T-PI operator, the PIE becomes
\begin{align*}
    \p{(\T v)}{t}{}(t,s) \hspace{-1mm} &= \hspace{-1mm} \sum_{k=1}^{d} (\C_{k}[v^{\otimes k}])(t,s) = \sum_{i=1}^{m} (\M_{c_{i}} \circ \Hop_{i} [v^{\otimes \alpha_{i}}])(\cdot) \\
    &= \hspace{-1mm} \sum_{i=1}^{m} c_{i}(s) (\T v)^{\alpha_{i0}}(\cdot) (\R_{1} v)^{\alpha_{i1}}(\cdot) \hspace{-1mm} \dots \hspace{-1mm} (\R_{n} v)^{\alpha_{in}}(\cdot)
\end{align*}
Also by Lem. \ref{lem:fr}, relations \eqref{eq:derivs} are leveraged to express the PIE in the same form as the PDE.
\begin{equation*}
    \p{(\T v)}{t}{}(t,s) = c(s)^{\top} \U \left(\T v, \p{(\T v)}{s}{}, \dots, \p{(\T v)}{s}{n} \right)
\end{equation*}
It follows that, for any $v_0 \in L_2[\Omega]$, $v$ is a classical solution to the PIE defined by $\{\T,\C\}$ if and only if $\T v_0 \in X_B[\Omega]$ and $\T v$ is a classical solution to the PDE defined by $\{c,B,A\}$.

Conversely, for $j \in \{0, 1, \dots, n\}$ the operators $\T, \R_{j} \in \Pi_{3}$ exist and are defined by Lem. \ref{lem:fr}. Using the relations \eqref{eq:derivs}, the PDE becomes
\begin{align*}
    (\T \p{u}{s}{n})_{t}(t,s) &= c(s)^{\top} \U( \T \p{u}{s}{n}, \R_{1} \p{u}{s}{n}, \dots, \R_{n} \p{u}{s}{n}) \\
    &= \sum_{i=1}^{m} c_{i}(s) (\T \p{u}{s}{n})^{\alpha_{i0}} (\R_{1} \p{u}{s}{n})^{\alpha_{i1}} \hspace{-5mm} \dots (\R_{n} \p{u}{s}{n})^{\alpha_{in}}
\end{align*}
Defining operators $\Hop_{i}$, $\M_{i}$, and $\C_{k}$, as in \eqref{eq:ck}, results in the PDE \eqref{eq:pde} expressed in the same form as the PIE \eqref{eq:pie}.
\begin{align*}
    (\T \p{u}{s}{n})_{t}(t,s) &= \sum_{i=1}^{m} (\M_{c_{i}} \circ \Hop_{i}[ \p{u}{s}{n}^{\otimes \alpha_{i}}])(t,s) = \C \Z_{d}(\p{u}{s}{n})
\end{align*}
It follows that for a given $u_{0} \in X_{B}[\Omega]$, $u$ is a classical solution to the PDE, defined by $\{c, B, A\}$, if and only if $(u_{0})_{s_{n}} \in L_{2}[\Omega]$ and $\p{u}{s}{n}$ is a classical solution to the PIE, defined by $\{\T,\C\}$. 
\end{proof}

\section{SOS Polynomials on a Distributed State}
Inspired by SOSTOOLS \cite{prajna2002introducing}, which builds on SOS polynomials for the analysis of polynomial ODEs, we introduce the notion of a distributed SOS polynomial. This class allows us to enforce positivity of distributed polynomials by definition and through equality constraints. In the subsequent section, we employ distributed SOS polynomials as LFs and formulate an SOS program to certify exponential stability of the equilibrium solution, $u \equiv 0$, on an $L_{2}$-ball of radius $r$.

\begin{definition} \label{def:SOSpoly}
    For a given domain $\Omega := [a,b]$, let the vector of degree $\bar{d} \in \N$ monomials in $s, \theta \in \Omega$ be denoted by $\U_{\bar{d}}(s, \theta) \in \Real^{\mu(\bar{d})}[s,\theta]$, where $\mu(\bar{d}):= \frac{1}{2}(\bar{d}+1)(\bar{d}+2)$. Define
    \begin{align*}
    (\hat{\U} x)(s) &:=
    \begin{bmatrix}
    \int_a^s \U_{\bar{d}}(s,\theta)x(\theta) \, d\theta \\
    \int_s^b \U_{\bar{d}}(s,\theta)x(\theta) \, d\theta
    \end{bmatrix}
    \in \Pi_{2}^{2\mu(\bar{d})}
    \end{align*}
    We say that $q:L_{2}[\Omega]\to\mathbb{R}$ is a distributed SOS polynomial in $x$, of degree $2d\in\N$, if there exists $\bar{d}\in\N$ and a real matrix $\bar{Q}=\bar{Q}'\succeq 0$ of appropriate dimension such that
    \begin{equation} \label{eq:distsospoly}
        \hspace{-2mm} q(x) := \inner{\Z_{d}(x)}{\Q \Z_{d}(x)}_{L_{2}}
    \end{equation}
    with $\Q$ defined in \eqref{eq:Qop}. In which case, we write $q \in \Sigma_{2d}$.
\end{definition}

\addtocounter{equation}{1}

It should be observed that $\hat{\U}^{\otimes i}: L_{2}[\Omega^{i}] \to L_{2}^{(2 \mu(\bar{d}))^{i}}[\Omega]$ is a T-PI operator, each matrix $\bar{Q}_{ij}$ has dimension ${(2 \mu(\bar{d}))^{i} \times (2 \mu(\bar{d}))^{j}}$, and $\Q$ is a self-adjoint operator since $\bar{Q} = \bar{Q}'$. Furthermore, the distributed SOS polynomial can be expressed as a sum of quadratic forms
\begin{equation*}
    q(x) = \sum_{i,j=1}^{d} \inner{x^{\otimes i}}{\Q_{ij} x^{\otimes j}}_{L_{2}}
\end{equation*}
where $\Q_{ij} = (\hat{\U}^{\otimes i})^{*} \circ \bar{Q}_{ij} \circ \hat{\U}^{\otimes j}: L_{2}[\Omega^{j}] \to L_{2}[\Omega^{i}]$. The adjoint of a T-PI operator and the compositions in $\Q_{ij}$ are all well-defined; however, we do not provide closed-form expressions to compute them. By parametrizing distributed SOS polynomials in a quadratic manner, as in Def.~\ref{def:SOSpoly}, it is clear that any $q \in \Sigma_{2d}$ satisfies $q(x)\geq 0$ for all $x \in L_{2}[\Omega]$. However, since the quadratic representation of polynomials is not unique, it must be converted to the linear form for enforcing equality constraints. As a result, the adjoint and compositions in \eqref{eq:Qop} will not need to be computed explicitly since distributed SOS polynomials will always be implemented in linear form. The remainder of this section shows how to perform this conversion, thereby also verifying that elements of $\Sigma_{2d}$ are indeed distributed polynomials as per Def.~\ref{def:dispoly}. The next definition shows how the individual quadratic terms can be evaluated by defining the ``vectorization" of a distributed SOS polynomial.

\begin{definition} \label{def:vecop}
    For a given domain $\Omega := [a,b]$ and operator $\hat{\Q}: L_{2}[\Omega^{j}] \to L_{2}[\Omega^{i}]$ of the form
    \begin{equation*}
        \hat{\Q} = \underbrace{(\hat{\U}_{L,1} \otimes \cdots \otimes \hat{\U}_{L,i})^{*}}_{:=\hat{\U}_{L}} \circ \hspace{1mm} \hat{Q} \circ \underbrace{(\hat{\U}_{R,1} \otimes \cdots \otimes \hat{\U}_{R,j})}_{:= \hat{\U}_{R}}
    \end{equation*}
    where $\hat{\U}_{L,\ell} \in \Pi_{2}^{m}$, $\hat{\U}_{R,r} \in \Pi_{2}^{n}$, and $\hat{Q} \in \Real^{m^{i} \times n^{j}}$. Define
    \begin{align*}
        \tnf{vec}(\hat{\Q}) &:= \sum_{p=1}^{m^{i}} \sum_{q=1}^{n^{j}} \hat{Q}_{pq} \tnf{vec} \left( (\hat{\U}_{L})_{p} \otimes (\hat{\U}_{R})_{q} \right)
    \end{align*}
\end{definition}

It should be observed that in this definition, $\hat{\U}_{L}, \hat{\U}_{R} \in \Pi_{T}$, and that $\tnf{vec}(\hat{\Q})$ may be computed directly from the components $\hat{Q}$, $\hat{\U}_{L}$ and $\hat{\U}_{R}$. Combining Def. \ref{def:vecop} and Lem. \ref{lem:vec}, each quadratic term can be converted to the standard form of a distributed polynomial (Def. \ref{def:dispoly}). This is summarised in the following Lemma.

\begin{lemma} \label{lem:quad_ind}
    For a given domain $\Omega := [a,b]$ and any $\hat{\Q}: L_{2}[\Omega^{j}] \to L_{2}[\Omega^{i}]$ as defined in Def. \ref{def:vecop}
    \begin{equation} \label{eq:lem_vecop}
        \inner{z}{\hat{\Q} y}_{L_{2}} = \tnf{vec}(\hat{\Q}) (z \otimes y) 
    \end{equation}
    $\fa z \in L_{2}[\Omega^{i}]$ and $y \in L_{2}[\Omega^{j}]$.
\end{lemma}

\begin{proof}
    By Def. \ref{def:vecop} and expanding the inner product, for any $z \in L_{2}[\Omega^{i}]$ and $y \in L_{2}[\Omega^{j}]$
    \begin{align*}
        \inner{z}{\hat{\Q} y}_{L_{2}} &= \inner{\hat{\U}_{L} z}{\hat{Q} \circ \hat{\U}_{R} y}_{L_{2}} \\
        &= \sum_{p=1}^{m^{i}} \sum_{q=1}^{n^{j}} \hat{Q}_{pq} \int_{\Omega} ((\hat{\U}_{L})_{p} z)(s) ((\hat{\U}_{R})_{q} y)(s) d s
    \end{align*}
    As $(\hat{\U}_{L})_{p}: L_{2}[\Omega^{i}] \to L_{2}[\Omega]$ and $(\hat{\U}_{R})_{q}: L_{2}[\Omega^{j}] \to L_{2}[\Omega]$ are both scalar T-PI operators, applying Def. \ref{def:TPI} leads to
    \begin{align*}
        &= \sum_{p=1}^{m^{i}} \sum_{q=1}^{n^{j}} \hat{Q}_{pq} \int_{\Omega} \left( \big( (\hat{\U}_{L})_{p}  \otimes (\hat{\U}_{R})_{q} \big) [z \otimes y]\right)(s) d s
    \end{align*}
    Applying Lem. \ref{lem:vec}
    \begin{align*}
        &= \sum_{p=1}^{m^{i}} \sum_{q=1}^{n^{j}} \hat{Q}_{pq} \tnf{vec} \big( (\hat{\U}_{L})_{p}  \otimes (\hat{\U}_{R})_{q} \big) [z \otimes y]
    \end{align*}
    followed by Def. \ref{def:vecop} results in \eqref{eq:lem_vecop}.
\end{proof}

This section concludes by converting the distributed SOS polynomial, defined in quadratic form, to the standard (linear) form of Def. \ref{def:dispoly}.

\begin{lemma} \label{lem:lin_ppi}
    For a given domain $\Omega := [a,b]$, a distributed SOS polynomial $q \in \Sigma_{2d}$, as defined in Def. \ref{def:SOSpoly}, satisfies
    \begin{align} \label{eq:SOSpoly_linear}
        q(x) = \tnf{vec}(\Q) \Z_{2d}(x)
    \end{align}
    where $\tnf{vec}(\Q) := [0 \hspace{2mm} \Q_{2} \hspace{2mm} \cdots \hspace{2mm} \Q_{2d} ]$ contains F-PI operators $\Q_{w}: L_{2}[\Omega^{w}] \to \Real$ of the form
    \begin{align*}
        \Q_{w} &= \sum_{(i,j) \in \Sm_{w}} \tnf{vec}(\Q_{ij})
    \end{align*}
    and the F-PI operators, $\tnf{vec}(\Q_{ij})$, are such that $\inner{x^{\otimes i}}{\Q_{ij} x^{\otimes j}}_{L_{2}} = \tnf{vec}(\Q_{ij}) (x^{\otimes i} \otimes x^{\otimes j})$ $\fa x \in L_{2}[\Omega]$.
\end{lemma}

\begin{proof}
    Starting from Def. \ref{def:SOSpoly} and applying Lem. \ref{lem:quad_ind}, with $\tnf{vec}(\Q_{ij})$ defined in Def. \ref{def:vecop},
    \begin{align*}
        q(x) &= \sum_{i,j=1}^{d} \tnf{vec} (\Q_{ij}) (x^{\otimes i} \otimes x^{\otimes j}) 
    \end{align*}
    Since F-PI operators are closed under linear combinations, operators acting on the same degree distributed monomial can be summed (i.e, summing over $\Sm_{w}$ as defined in Section \ref{sec:not}) to get \eqref{eq:SOSpoly_linear}.
\end{proof}

Armed with this Lemma, we can now define a LF as a distributed SOS polynomial and enforce negativity of its derivative through an equality constraint, as shown next.

\section{Distributed SOS Program For Local Stability Analysis of Polynomial PDEs} \label{sec:local}
In the previous sections, we have shown how polynomial PIEs and semi-algebraic sets can be represented as distributed polynomials. Furthermore, we have also parameterized a class of nonnegative distributed polynomials, as distributed SOS polynomials. Using these results, we now show how local stability of the equilibrium $u \equiv 0$ for a PDE as in~\eqref{eq:pde} can be verified in the associated PIE representation~\eqref{eq:pie}, by solving a distributed SOS program. Specifically, to verify stability, we look for a distributed SOS polynomial LF, $V(v)$, which decays towards zero on a given semi-algebraic region. We test this decay by verifying negativity of the Lie derivative, defined as follows.

\begin{definition}
    For a given polynomial PIE defined by $\{\T,\C\}$ and distributed polynomial, $V(v)$, define the Lie derivative $\mathcal{L}_{\T,\C}V(v)$ such that $\mathcal{L}_{\T,\C}V(v(t))=\frac{d}{dt}V(v(t))$ for any classical solution $v(t)$ to the PIE. 
\end{definition}

Choosing $V \in \Sigma_{2d}$, the Lie derivative $\mathcal{L}_{\T,\C}V(v)$ along the polynomial PIE in~\eqref{eq:pie} will also be a distributed polynomial. In order to evaluate this derivative, we assume $V$ to be of the form $V(v)=\hat{V}(\T v)=\hat{V}(u)$, where $\hat{V}(u)$ is then a LF for the PDE. The following Lemma provides an explicit expression for the Lie derivative of a degree 2 distributed SOS LF. 
Although the Lie derivative of a higher-degree distributed SOS LF may be computed similarly, this derivation is significantly more involved, and is therefore deferred to a future publication.

\begin{lemma}
    For a given polynomial PIE defined by $\{\T,\C\}$, let $\C:=\bmat{\C_{1}&\cdots&\C_{d'}}$. For a self-adjoint 3-PI operator $\Pop=\Pop^*$, let $V(v):=\inner{v}{\T^{*}\Pop\T v}_{L_{2}}$. Then
    \begin{equation*}
        \mathcal{L}_{\T,\C}V(v)
        =2\sum_{i=1}^{d'}\tnf{vec} \left(\C_{i} \otimes (\Pop \circ \T) \right)[v^{\otimes i+1}]
    \end{equation*}
\end{lemma}
\begin{proof}
    Let $v(t)$ be an arbitrary classical solution to the PIE defined by $\{\T,\C\}$. Then, by Def. \ref{def:TPI} and Lem. \ref{lem:vec},
    \begin{align*}
        \frac{d}{dt}V(v(t))
        &=2\inner{\T v_{t}(t)}{\Pop\T v(t)}_{L_{2}} \\
        &=2\inner{\C \Z_{d}(v(t))}{\Pop\T v(t)}_{L_{2}} \\
        &=2\sum_{i=1}^{d'}\int_{\Omega}(\C_{i}v(t,\cdot)^{\otimes i})(s)(\Pop\T v(t,\cdot))(s)\, ds  \\
        &=2\sum_{i=1}^{d'}\int_{\Omega}\left(\C_{i} \otimes (\Pop \circ \T)\right)[v(t,\cdot)^{\otimes i+1}](s)\, ds    \\
        &=2\sum_{i=1}^{d'}\tnf{vec}(\C_{i} \otimes (\Pop \circ \T) )v(t)^{\otimes i+1} \tag*{\qed}
    \end{align*}
    \renewcommand{\qedsymbol}{}
\end{proof}
Given that the Lie derivative of a distributed polynomial is again a distributed polynomial, negativity of this derivative can be tested by searching for a distributed SOS polynomial, $W\in\Sigma_{2d}$, such that $\mathcal{L}_{\T,\C}V=-W$. In this manner, existence of a LF certifying stability on the $L_{2}$-ball of radius $r$ can be formulated as a distributed SOS program, yielding the following Theorem.

\begin{theorem} \label{thm:stabilitySOS}
    For a given polynomial PDE with domain $X_{B}$ as in~\eqref{eq:X_B} satisfying the conditions of Lem.~\ref{lem:fr}, let $\{\T,\C\}$ denote the associated polynomial PIE. For $r > 0$, let $g_{r}(v) := r^{2} - \tnf{vec}(\T \otimes \T) v^{\otimes 2}$. For some $\epsilon,C >0$, $\lambda\geq 0$, and $d, d_{1}', d_{2}' \in \N$, suppose there exists $V \in \Sigma_{2d}$, $p_{1} \in \Sigma_{2d_{1}'}$, and $p_{2} \in \Sigma_{2d_{2}'}$, such that
    \begin{align*}
        & & V(v) - \epsilon^2 \tnf{vec}(\T \otimes \T)v^{\otimes 2} & \in\Sigma_{2d} \\
        & & C^2 \tnf{vec}(\T \otimes \T)v^{\otimes 2} -V(v) - p_{1}(v) g_{r}(v) & \in\Sigma_{2\hat{d}_{1}} \\
        & & -\mathcal{L}_{\T,\C}V(v) -2\lambda V(v) - p_{2}(v) g_{r}(v) & \in \Sigma_{2\hat{d}_{2}}
    \end{align*}
    where $\hat{d}_{i} = \max \{d, d_{i}'+1\}$. Then, for $M:=C/\epsilon$ and any $u_{0} \in X_{B,r/M}$, as defined in \eqref{eq:Xbr}, solutions to the PDE with initial state $u(0)=u_{0}$ satisfy
    \begin{equation*}
        \norm{u(t)}{L_{2}}{} \leq Me^{-\lambda t} \norm{u_{0}}{L_{2}}{} \hspace{1mm} \fa t\geq 0
    \end{equation*}
\end{theorem}
\begin{proof}
    Suppose there exists $V$ satisfying the proposed conditions, and let $\hat{V}(u):=V(u_{s_{n}})$ for $u\in X_{B}$, so that $\hat{V}(\T v)=V(v)$. An $\epsilon >0$ must exist such that $\hat{V}(u) \geq \epsilon^2 \|u\|_{L_{2}}^{2}$; hence, a sufficiently small $\gamma > 0$  can always be chosen such that $\hat{V}(u) \leq \gamma^2$, which implies $\|u\|_{L_{2}}^2 \leq r^{2}$. That is, the $\gamma^2$-level set of $\hat{V}(u)$ is contained in $X_{B,r}$. Now define $L_{2,r}[\Omega] := \{v\in L_{2}[\Omega]: \hspace{1mm} g_{r}(v) \geq  0 \}$ and suppose the distributed SOS program is feasible. Then, for all $v \in L_{2,r}$
    \begin{align*}
        V(v) \geq \epsilon^2 \tnf{vec}(\T \otimes \T)v^{\otimes 2} & \geq \epsilon^2 \norm{\T v}{L_{2}}{2} \\
        V(v) \leq C^2 \tnf{vec}(\T \otimes \T)v^{\otimes 2} - p_{1}(v)g_{r}(v) & \leq C^2 \norm{\T v}{L_{2}}{2} \\
        \mathcal{L}_{\T,\C}V(v) \leq -2\lambda V(v) - p_{2}(v)g_{r}(v) &\leq 0
    \end{align*}
    For $\gamma>0$, define the $\gamma^2$-sublevel set of $V$ as
    \begin{equation*}
        S_{\gamma}(V) := \{v\in L_{2}[\Omega]: \hspace{1mm} V(v) \leq \gamma^2\}
    \end{equation*}
    Since $V(v) \geq \epsilon^2 \norm{\T v}{L_{2}}{2}$ for all $v \in L_{2}[\Omega]$, it follows that for all $v \in S_{\epsilon r}$
    \begin{equation*}
        \norm{\T v}{L_{2}}{2} \leq \frac{1}{\epsilon^2} V(v) \leq r^2
    \end{equation*}
    whence $S_{\epsilon r} \subseteq L_{2,r}$. Since $\dot{V}(v) \leq 0$ for all $v \in L_{2,r}$, this implies that for any solution $v(t)$ to the PIE, $v(0)\in S_{\epsilon r}$ implies $v(t)\in L_{2,r}$ for all $t \geq 0$. Furthermore, since $V(v) \leq C^2 \norm{\T v}{L_{2}}{2}$, we have for any $v_{0} \in L_{2,r/M}$ that
    \begin{equation*}
        V(v_{0}) \leq C^2 \norm{\T v}{L_{2}}{2} \leq r^2 \epsilon^2
    \end{equation*}
    whence $L_{2,r/M} \subseteq S_{\epsilon r}$. It follows that, if $v(0) = v_{0} \in L_{2,r/M}$, then $v(t)\in S_{\epsilon r} \subseteq L_{2,r}$ and therefore $\mathcal{L}_{\T,\C}V(v(t)) \leq -2\lambda V(v(t))$ for all $t \geq 0$. By the Gr\"onwall-Bellman inequality, $v_{0}\in L_{2,r/M}$ implies
    \begin{equation*}
        V(v(t)) \leq e^{-2 \lambda t} V(v_{0}) \hspace{1mm} \fa t \geq 0
    \end{equation*}
    Now, for any $u_{0}\in X_{B,r/M}$, let $u(t)\in X_{B}$ be the associated solution to the PDE for $t \geq 0$. Then $u(t) = \T v(t)$ for all $t \geq 0$, where $v(t) := \p{u}{s}{n}(t) \in L_{2}[\Omega]$ is a solution to the PIE, with $v_{0} := (u_{0})_{s_{n}}$. Since $u_{0} \in X_{B,r/M}$, we have $\norm{u_{0}}{L_{2}}{} = \norm{\T v_{0}}{L_{2}}{} \leq r/M$, and therefore $v_{0} \in L_{2,r/M}$. It follows that, for all $t \geq 0$
    \begin{align*}
        \norm{u(t)}{L_{2}}{} &= \norm{\T v(t))}{L_{2}}{} \leq \frac{1}{\epsilon}\sqrt{V(v(t))} \\
        &\leq \frac{1}{\epsilon} e^{-\lambda t} \sqrt{V(v_{0})} \\
        &\leq \frac{C}{\epsilon} e^{- \lambda t} \norm{\T v_{0}}{L_{2}}{} = Me^{- \lambda t}\norm{u_{0}}{L_{2}}{}  \tag*{\qed}
    \end{align*}
    \renewcommand{\qedsymbol}{}
\end{proof}

\begin{table}[t!]
\setlength{\tabcolsep}{3.75pt}
    \caption{Largest lower bound on decay rate $\lambda$ ensuring exponential stability of the Fisher Eq. \eqref{eq:Fisher} on an $L_2$-ball of radius $r$.}
    \begin{center}
        \vspace{-4mm}
        \begin{tabular}{c|ccccccccc}
             $r$ & 0.01 & 0.05 & 0.1 & 0.25 & 0.5 & 1 & 2 & 3 & 4 \\\hline
             $\lambda_{r}$ & 4.857 & 4.810 & 4.734 & 4.572 & 4.270  & 3.650 & 2.692 & 1.388 & 0.047 
        \end{tabular}
        \vspace{-3mm}
    \end{center}
     \label{tab:FisherRates}
\end{table}

\addtolength{\textheight}{-3.6cm}

\section{Numerical Example} \label{sec:num}
In this section, we apply the distributed SOS program (Thm. \ref{thm:stabilitySOS}) to the Fisher Equation with Dirichlet BCs. 
\begin{equation} \label{eq:Fisher}
    \begin{split}
        u_{t}(t,s) &= u_{ss}(t,s) + \alpha u(t,s) - \beta u(t,s)^2 \\
        u(t,0) &= u(t,1) = 0 \hspace{3mm} t \geq 0\, \hspace{1mm} s\in[0,1]
    \end{split}
\end{equation}
%
We consider this system with parameters $(\alpha, \beta) =(5, -1)$, where the origin of the linearization is exponentially stable with rate $\lambda=4.869$. However, for the nonlinear system, the term $-u(t,s)^2$ dominates when $\|u(t)\|_{L_{2}}$ becomes too large, so the origin is only locally stable. To verify this numerically, we construct the PIE representation as
\begin{equation*}
    \T \p{v}{t}{}(t,s) = (\C_{1}v)(t,s) + (\C_{2}[v^{\otimes 2}])(t,s)
\end{equation*}
where $v:=u_{ss}$, $\C_{1}=1 + 5\T \in \Pi_{3}$, $C_{2} = - \T^{\otimes 2}$, and
\begin{equation*}
    (\T v)(s):=\int_{0}^{s}(s-1)\theta v(\theta)\, d\theta +\int_{s}^{1}s(\theta-1)v(\theta)\, d\theta
\end{equation*}

\begin{figure}[t!]
    \centering
    \includegraphics[width=0.99\linewidth]{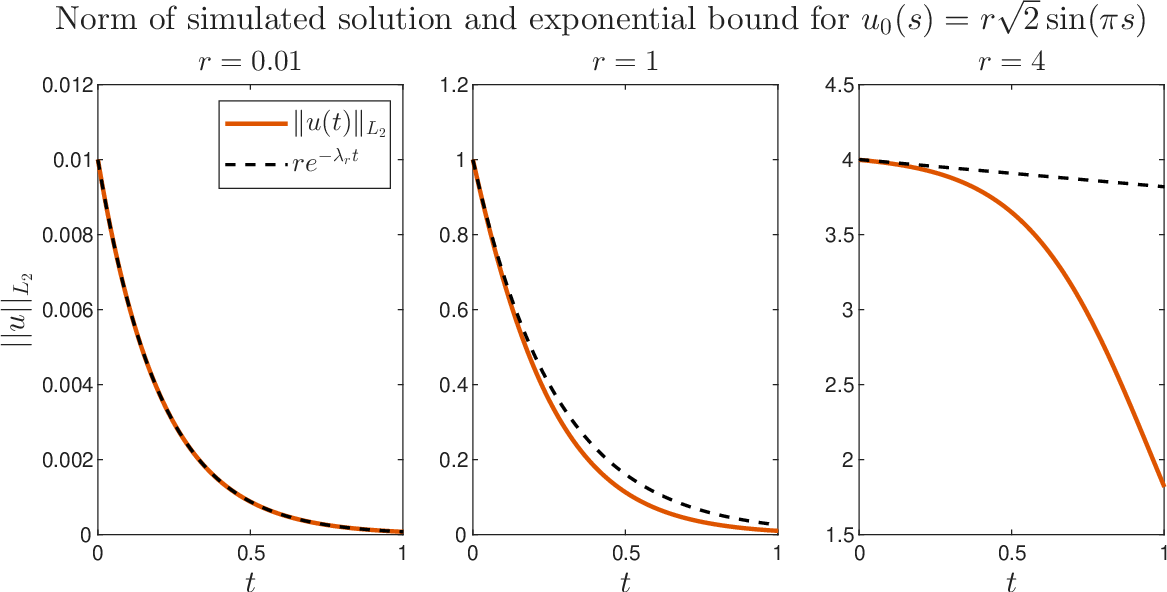}
    \vspace{-7mm}
    \caption{Norm of simulated solutions to the Fisher Equation~\eqref{eq:Fisher} for $u_{0}(s):=r\sqrt{2}\sin(\pi s)$, along with an exponential upper bound $re^{-\lambda_{r} t}$ with decay rates $\lambda_{r}$ as in Table~\ref{tab:FisherRates}.}
    \label{fig:FisherSim}
    \vspace{-2mm}
\end{figure}

We can verify local stability of the Fisher Equation by solving the distributed SOS program in Thm.~\ref{thm:stabilitySOS}. Specifically, we fix $\lambda=0$ and declare a quadratic LF, $V(v)=\inner{\T v}{\Q \T v}_{L_{2}} \in \Sigma_{2}$, where $\Q=\hat{\Q}+\epsilon^2 I$ for $\hat{Q}$ as in Def.~\ref{def:SOSpoly} with $d=1$ and $\bar{d}=4$. Performing bisection on $r$, feasibility of the distributed SOS program (Thm.~\ref{thm:stabilitySOS}) can then be verified up to $r=4.048$. In particular, feasibility can be verified with the energy functional $V(v)=\inner{\T v}{\T v}_{L_{2}}$ (where $M=\epsilon$). Simulating the PDE with initial condition $u_{0}(s) = r \sqrt{2} \sin(\pi s)$ (so that $\|u_{0}\|_{L_{2}}=r$), instability is observed for $r>4.05$, suggesting the lower bound is tight. 

Thm.~\ref{thm:stabilitySOS} is also applied to verify exponential stability of the PDE~\eqref{eq:Fisher} on the $L_{2}$-ball of varying radius $r$, in the PDE domain. The largest lower bound, $\lambda$, on the exponential decay rate is established by performing bisection on $\lambda$. The results are given in Table \ref{tab:FisherRates}. For each value of $r$ and $\lambda$, the obtained value of $M$ in Thm.~\ref{thm:stabilitySOS} is roughly equal to 1, certifying that $\|u_{0}\|_{L_{2}}\leq r$ implies $u(t) \leq e^{-\lambda t} \|u_{0}\|_{L_{2}}$ for all $t \geq 0$. The simulated solutions for $u_{0}(s) = r \sqrt{2} \sin(\pi s)$ with $r \in \{0.01,1.0,4.0\}$ are displayed in Fig.~\ref{fig:FisherSim}, along with the associated exponential bounds $re^{-\lambda t}$. The results show that the obtained $\lambda$ provides a tight lower bound on the initial decay rate of the solution starting with $\|u_{0}\|_{L_{2}} \leq r$. As the value of $r$ decreases to 0, the lower bound $\lambda$ also increases to the rate $\lambda=4.869$ of the linearized system, as expected.

\section{Conclusion}

In this paper, we established a new framework for the local stability analysis of 1D polynomial PDEs with linear BCs. By leveraging the PIE representation, we transformed the evolution of a PDE state into the evolution of its highest spatial derivative — the ``fundamental state". This shift effectively embeds BCs and continuity constraints directly into the operator algebra, removing the need to enforce them as external constraints on the state space. Our primary technical contributions include:
\begin{itemize}
    \item Distributed Polynomials - A unifying representation of polynomial PDEs, LFs, and distributed semi-algebraic sets.
    \item Distributed SOS Programming - A SOS parameterization of distributed polynomials, which allowed us to enforce positivity of distributed polynomials through equality constraints on the operator parameters.
    \item Local Stability Certificates - We constructed a distributed SOS program to test for local exponential stability on the $L_{2}$-ball, using distributed SOS polynomials as LF candidates.
\end{itemize}
This approach offers a systematic and computationally tractable method for analysing a wide range of nonlinear PDEs, without the application-specific manual derivations required by prior methods. Future research aims to extend this work to coupled PDEs with multiple spatial dimensions by combining with results from \cite{jagt2025state}. Furthermore, integrating these methods into the PIETOOLS software package will provide a general-purpose tool for the automated control and stability analysis of nonlinear PDEs.

\appendices
\section{Proofs}

\subsection{Proof of Lemma \ref{lem:int_split}} \label{app:int_split}
\begin{proof}
    The result follows by induction. For the base case $d=1$, the result holds trivially since $\mathbb{I}_{1}(\theta_{1})=1$. Now, suppose the relation holds for some $d \in \N$. Then, for $d+1$
    \begin{align*}
        &\int_{\Omega^{d+1}} \hspace{-5mm} K(\theta_{d+1}, \bs{\theta}_{d}) \, d\theta_{d+1} d \bs{\theta}_{d} = \int_{\Omega^{d}} \int_{\Omega} \hspace{-1mm} K(\theta_{d+1}, \bs{\theta}_{d}) \, d\theta_{d+1} d \bs{\theta}_{d}\\
        &= \sum_{\alpha\in A_{d}} \int_{\Omega^{d}} \int_{\Omega} \mathbb{I}_{\alpha}(\bs{\theta}_{d}) K(\theta_{d+1}, \bs{\theta}_{d})\, d \theta_{d+1} d \bs{\theta}_{d} \\
        &= \sum_{\alpha\in A_{d}} \int_{\Omega^{d}} \mathbb{I}_{\alpha}(\bs{\theta}_{d}) \left[ \int_{a}^{\theta_{\alpha_{1}}} K(\theta_{d+1}, \bs{\theta}_{d})\, d\theta_{d+1} \right. \\
        &\hspace*{2.0cm} +\sum_{j=1}^{d-1} \int_{\theta_{\alpha_{j}}}^{\theta_{\alpha_{j+1}}} K(\theta_{d+1}, \bs{\theta}_{d})\, d\theta_{d+1} \\
        &\hspace*{3.0cm} + \left. \int_{\theta_{\alpha_{d}}}^{b} K(\theta_{d+1}, \bs{\theta}_{d})\, d\theta_{d+1} \right] d \bs{\theta}_{d} \\
        &=\sum_{\alpha'\in A_{d+1}} \int_{\Omega^{d+1}} \mathbb{I}_{\alpha'}(\bs{\theta}_{d+1}) K(\bs{\theta}_{d+1}) \, d \bs{\theta}_{d+1}  \tag*{\qed}
    \end{align*}
    \renewcommand{\qedsymbol}{}
\end{proof}

\subsection{Proof of Lemma \ref{lem:polynomialProduct}}\label{app:polynomialProduct}

\begin{proof}
    Fix arbitrary $i,j\in\{1,\hdots,d\}$, and let $\K_{i}:L_{2}[\Omega^{i}]\to\Real$ and $\G_{j}:L_{2}[\Omega^{j}]\to\Real$ be defined by parameters $\{K_{\alpha}\}_{\alpha\in A_{i}}$ and $\{G_{\beta}\}_{\beta\in A_{j}}$, respectively. Then, for all $x\in L_{2}[\Omega]$,
    \begin{align*}
        &(\K_{i}x^{\otimes i})\, (\G_{j}x^{\otimes j}) \\
        &=\hspace{-2mm} \sum_{\alpha\in A_{i}} \hspace{-1mm} \sum_{\beta\in A_{j}} \hspace{-2mm} \int_{\Omega^{i}} \hspace{-1mm}  \int_{\Omega^{j}} \hspace{-3mm} \mbb{I}_{\alpha}(\bs{\theta}_{i})\mbb{I}_{\beta}(\bs{\eta}_{j})H_{\alpha,\beta}(\bs{\theta}_{i},\bs{\eta}_{j}) w(\bs{\theta}_{i},\bs{\eta}_{j})\, d\bs{\eta}_{j} d\bs{\theta}_{i}
    \end{align*}
    where we define $H_{\alpha,\beta}(\bs{\theta}_{i}, \bs{\eta}_{j}):=K_{\alpha}(\bs{\theta}_{i})G_{\beta}(\bs{\eta}_{j})$ and $w \in L_{2}[\Omega^{i+j}]$.
    Now, for each $\alpha\in A_{i}$ and $\beta\in A_{j}$, define $\Xi$ as the set of all unions $\alpha\cup(\beta+i)$ that preserve the ordering of $\alpha$ and $\beta$,
    \begin{equation}
        \Xi_{\alpha,\beta}:=\{\gamma\in A_{i+j}\mid \sigma_{1}(\gamma)=\alpha,~\sigma_{2}(\gamma)-i=\beta\},
    \end{equation}
    where $\sigma_{1}(\gamma):=\gamma\setminus\{i+1,\hdots,i+j\}\in A_{i}$ and $\sigma_{2}(\gamma):=\gamma\setminus\{1,\hdots,i\}-i\in A_{j}$. Then, for any $\gamma\in A_{i+j}$,
    \begin{equation*}
        \mbb{I}_{\gamma}(\bs{\theta}_{i},\bs{\eta}_{j})\mbb{I}_{\alpha}(\bs{\theta}_{i})\mbb{I}_{\beta}(\bs{\eta}_{j})=
        \begin{cases}
            \mbb{I}_{\gamma}(\bs{\theta}_{i},\bs{\eta}_{j})    & \text{if} \hspace{2mm} \gamma\in \Xi_{\alpha,\beta} \\
            0  &   \tnf{else}
        \end{cases}
    \end{equation*}
    and therefore, by Lem.~\ref{lem:int_split},
    \begin{align*}
        &\int_{\Omega^{i}}\int_{\Omega^{j}}\mbb{I}_{\alpha}(\bs{\theta}_{i})\mbb{I}_{\beta}(\bs{\eta}_{j})H_{\alpha,\beta}(\bs{\theta}_{i},\bs{\eta}_{j})w(\bs{\theta}_{i},\bs{\eta}_{j})\, d\bs{\eta}_{j} d\bs{\theta}_{i}  \\
        &=\sum_{\gamma\in A_{i+j}}\int_{\Omega^{i}}\int_{\Omega^{j}}\mbb{I}_{\gamma}(\bs{\theta}_{i},\bs{\eta}_{j})\mbb{I}_{\alpha}(\bs{\theta}_{i})\mbb{I}_{\beta}(\bs{\eta}_{j})\\
        &\hspace*{4.0cm}H_{\alpha,\beta}(\bs{\theta}_{i},\bs{\eta}_{j})w(\bs{\theta}_{i},\bs{\eta}_{j})\, d\bs{\eta}_{j} d\bs{\theta}_{i}\\        
        &=\sum_{\gamma\in \Xi_{\alpha,\beta}}\int_{\Omega^{i+j}}\mbb{I}_{\gamma}(\bs{\theta}_{i+j})H_{\sigma_{1}(\gamma),\sigma_{2}(\gamma)}(\bs{\theta}_{i+j})w(\bs{\theta}_{i+j})\, d\bs{\theta}_{i+j}
    \end{align*}
    Noting that $\bigcup_{\alpha\in A_{i}}\bigcup_{\beta\in A_{j}}\Xi_{\alpha,\beta}=A_{i+j}$, this implies
    \begin{align*}
        &(\K_{i}x^{\otimes i})\, (\G_{j}x^{\otimes j})
        =\sum_{\alpha\in A_{i}}\sum_{\beta\in A_{j}}
        \sum_{\gamma\in \Xi_{\alpha,\beta}}\int_{\Omega^{i+j}}\mbb{I}_{\gamma}(\bs{\theta}_{i+j})\\
        &\hspace*{3.0cm}H_{\sigma_{1}(\gamma),\sigma_{2}(\gamma)}(\bs{\theta}_{i+j})x^{\otimes i+j}(\bs{\theta}_{i+j})\, d\bs{\theta}_{i+j}\\
        &=\sum_{\gamma\in A_{i+j}}\int_{\Omega^{i+j}}\mbb{I}_{\gamma}(\bs{\theta}_{i+j})H_{\sigma_{1}(\gamma),\sigma_{2}(\gamma)}(\bs{\theta}_{i+j})x^{\otimes i+j}(\bs{\theta}_{i+j})\, d\bs{\theta}_{i+j} \\
        &=(\K_{i}\otimes \G_{j})x^{\otimes i+j}.
    \end{align*}    
    Now, let $p(x)=\K\Z_{d}(x)$ and $q(x)=\G\Z_{d}(x)$. Then, it follows that
    \begin{align*}
        p(x)q(x)
        &=\K\Z_{d}(x)\, \G\Z_{d}(x) =\sum_{i=1}^{d}\sum_{j=1}^{d}(\K_{i}x^{\otimes i})(\G_{j} x^{\otimes j})   \\
        &=\sum_{i=1}^{d}\sum_{j=1}^{d}(\K_{i}\otimes\G_{j})x^{\otimes i+j}
        =\C\Z_{2d}(x).
    \end{align*}
    where $\C = [0 \hspace{2mm} \C_{2} \hspace{2mm} \cdots \hspace{2mm} \C_{2d}]$, $\C_{k} = \sum_{(i,j) \in \Sm_{k}} \K_{i} \otimes \G_{j}$, and $\Sm_{k}$ is defined in Section \ref{sec:not}.
\end{proof}

\subsection{Proof of Lemma \ref{lem:vec}} \label{app:vec}
\begin{proof}
    Let $\Omega:= [a,b]$. For $i \in \{1, \dots ,d\}$ and a given set of polynomial parameters $R_{i,1}, R_{i,2} \in \Real[s,\theta_{i}]$, define the set of scalar 2-PI operators $\R_{i} = \Pop_{ \{ R_{i,1}, R_{i,2} \} } \in \Pi_{2}$. Now, let $\Hop = \R_{1} \otimes \cdots \otimes \R_{d}: L_{2}[\Omega^{d}] \to L_{2}[\Omega]$. Then,
    \begin{align*}
        &\int_{\Omega} (\Hop x)(s) ds =  \int_{\Omega} (\R_{1} x_{1})(s) \cdots (\R_{d} x_{d})(s) \, ds \\
        &= \sum_{\beta \in \{1,2\}^{d}} \int_{\Omega} \hspace{-1mm} \int_{\Omega^{d}} \prod_{i=1}^{d} \left[\mbf{I}^{\beta_{i}}(s,\theta_{i}) R_{i,\beta_{i}}(s,\theta_{i}) x_{i}(\theta_{i}) d \theta_{i} \right] ds        
    \end{align*}
    where $\mbf{I}^{1}(s,\theta_{i}) = \mbf{I}_{0^{+}}(s - \theta_{i})$, $\mbf{I}^{2}(s,\theta_{i}) = \mbf{I}_{0^{+}}(\theta_{i} - s)$. Defining $x(\bs{\theta}_{d}):= x_{d}(\theta_{d}) \cdots x_{1}(\theta_{1})$, and applying Lemma \ref{lem:int_split}, where $K(\bs{\theta}_{d}) = M(\bs{\theta}_{d}) x(\bs{\theta}_{d})$, it follows that
    \begin{align*}
        &= \int_{\Omega^{d}} \Big[ \hspace{-2mm} \overbrace{\sum_{\beta\in\{1,2\}^{d}} \int_{\Omega} \prod_{i=1}^{d} \mbf{I}^{\beta_{i}}(s,\theta_{i}) R_{i,\beta_{i}}(s,\theta_{i}) ds \Big]}^{:= M(\bs{\theta}_{d})}  x(\bs{\theta}_{d}) \, d \bs{\theta}_{d} \\
        &= \sum_{\alpha\in A_{d}} \int_{\Omega^{d}} \mathbb{I}_{\alpha}(\bs{\theta}_{d}) M(\bs{\theta}_{d}) x(\bs{\theta}_{d}) \, \hspace{-1mm} d \bs{\theta}_{d}
    \end{align*}
    Now, consider an arbitrary $\alpha\in A_{d}$. Then, we can decompose
    \begin{align*}
        &\mathbb{I}_{\alpha}(\bs{\theta}_{d}) M(\bs{\theta}_{d}) \\
        &= \sum_{\beta\in\{1,2\}^{d}} \sum_{j=0}^{d} \int_{\theta_{j}}^{\theta_{j+1}} \mathbb{I}_{\alpha}(\bs{\theta}_{d})\prod_{i=1}^{d} [\mbf{I}^{\beta_{i}}(s,\theta_{i}) R_{i,\beta_{i}}(s,\theta_{i}) ] ds
    \end{align*}
    where $\theta_{0}:=a$ and $\theta_{d+1}:=b$. By definition of $\gamma(\alpha_{k},j)$ (Section \ref{sec:not}) and $\mbf{I}^{\beta_{i}}(s,\theta_{i})$, it follows that
    \begin{align*}
        \mathbb{I}_{\alpha}(\bs{\theta}_{d}) M(\bs{\theta}_{d}) \hspace{-1mm} &= \hspace{-1mm} \sum_{j=0}^{d} \int_{\theta_{\alpha_{j}}}^{\theta_{\alpha_{j+1}}} \hspace{-1mm} \prod_{i=1}^{d} R_{i,\gamma(i,j)}(s,\theta_{i}) \, ds := Q_{\alpha}(\bs{\theta}_{d})
    \end{align*}
    Subbing $Q_{\alpha}(\bs{\theta}_{d})$ back into the full expression results in
    \begin{align*}
        \int_{\Omega} (\Hop x)(s) ds &= \sum_{\alpha \in A_{d}} \int_{\Omega^{d}} Q_{\alpha}(\bs{\theta}_{d}) x(\bs{\theta}_{d}) d \bs{\theta}_{d} \\
        &= \tnf{vec}(\Hop) x \quad \fa x \in L_{2}[\Omega^{d}]  \tag*{\qed}
    \end{align*}
    \renewcommand{\qedsymbol}{}
\end{proof}

\bibliographystyle{ieeetr}
\bibliography{refs} 

@inproceedings{valmorbida2014semi,
  title={Semi-definite programming and functional inequalities for distributed parameter systems},
  author={Valmorbida, Giorgio and Ahmadi, Mohamadreza and Papachristodoulou, Antonis},
  booktitle={53rd IEEE conference on decision and control},
  pages={4304--4309},
  year={2014},
  organization={IEEE}
}

@article{prieur2026stability,
  title={Stability of the wave equation with a saturated dynamic boundary control},
  author={Prieur, Christophe and Tarbouriech, Sophie},
  journal={Automatica},
  volume={187},
  pages={112920},
  year={2026},
  publisher={Elsevier}
}

@article{goulart2012global,
  title={Global stability analysis of fluid flows using sum-of-squares},
  author={Goulart, Paul J and Chernyshenko, Sergei},
  journal={Physica D: Nonlinear Phenomena},
  volume={241},
  number={6},
  pages={692--704},
  year={2012},
  publisher={Elsevier}
}

@article{huang2015sum,
  title={Sum-of-squares of polynomials approach to nonlinear stability of fluid flows: an example of application},
  author={Huang, Deqing and Chernyshenko, Sergei and Goulart, Paul and Lasagna, Davide and Tutty, Owen and Fuentes, Federico},
  journal={Proceedings of the Royal Society A: Mathematical, Physical and Engineering Sciences},
  volume={471},
  number={2183},
  year={2015},
  publisher={The Royal Society}
}

@article{fuentes2022global,
  title={Global stability of fluid flows despite transient growth of energy},
  author={Fuentes, Federico and Goluskin, David and Chernyshenko, Sergei},
  journal={Physical Review Letters},
  volume={128},
  number={20},
  pages={204502},
  year={2022},
  publisher={APS}
}

@article{goluskin2019bounds,
  title={Bounds on mean energy in the {K}uramoto--{S}ivashinsky equation computed using semidefinite programming},
  author={Goluskin, David and Fantuzzi, Giovanni},
  journal={Nonlinearity},
  volume={32},
  number={5},
  pages={1705--1730},
  year={2019},
  publisher={IOP Publishing}
}

@article{fridman2016new,
  title={New stability and exact observability conditions for semilinear wave equations},
  author={Fridman, Emilia and Terushkin, Maria},
  journal={Automatica},
  volume={63},
  pages={1--10},
  year={2016},
  publisher={Elsevier}
}

@article{valmorbida2015stability,
  title={Stability analysis for a class of partial differential equations via semidefinite programming},
  author={Valmorbida, Giorgio and Ahmadi, Mohamadreza and Papachristodoulou, Antonis},
  journal={IEEE Transactions on Automatic Control},
  volume={61},
  number={6},
  pages={1649--1654},
  year={2015},
  publisher={IEEE}
}

@article{mironchenko2019monotonicity,
  title={Monotonicity methods for input-to-state stability of nonlinear parabolic {PDEs} with boundary disturbances},
  author={Mironchenko, Andrii and Karafyllis, Iasson and Krstic, Miroslav},
  journal={SIAM Journal on Control and Optimization},
  volume={57},
  number={1},
  pages={510--532},
  year={2019},
  publisher={SIAM}
}

@incollection{korda2022moments,
  title={Moments and convex optimization for analysis and control of nonlinear {PDEs}},
  author={Korda, Milan and Henrion, Didier and Lasserre, Jean Bernard},
  booktitle={Handbook of Numerical Analysis},
  volume={23},
  pages={339--366},
  year={2022},
  publisher={Elsevier}
}

@article{marx2018moment,
  title={A moment approach for entropy solutions to nonlinear hyperbolic {PDEs}},
  author={Marx, Swann and Weisser, Tillmann and Henrion, Didier and Lasserre, Jean},
  journal={arXiv preprint arXiv:1807.02306},
  year={2018}
}

@inproceedings{prajna2002introducing,
  title={Introducing {SOSTOOLS}: A general purpose sum of squares programming solver},
  author={Prajna, Stephen and Papachristodoulou, Antonis and Parrilo, Pablo A},
  booktitle={Proceedings of the 41st IEEE Conference on Decision and Control, 2002.},
  volume={1},
  pages={741--746},
  year={2002},
  organization={IEEE}
}

@article{shivakumar2025pietools,
  title={{PIETOOLS} 2024: User manual},
  author={Shivakumar, Sachin and Jagt, Declan and Braghini, Danilo and Das, Amritam and Peet, Yulia and Peet, Matthew},
  journal={arXiv preprint arXiv:2501.17854},
  year={2025}
}

@article{shivakumar2024extension,
  title={Extension of the partial integral equation representation to {GPDE} input--output systems},
  author={Shivakumar, Sachin and Das, Amritam and Weiland, Siep and Peet, Matthew},
  journal={IEEE Transactions on Automatic Control},
  volume={70},
  number={5},
  pages={3240--3255},
  year={2024},
  publisher={IEEE}
}

@inproceedings{jagt20222,
  title={L2-Gain Analysis of Coupled Linear {2D PDEs} using Linear {PI} Inequalities},
  author={Jagt, Declan S and Peet, Matthew M},
  booktitle={2022 IEEE 61st Conference on Decision and Control (CDC)},
  pages={6097--6104},
  year={2022},
  organization={IEEE}
}

@article{wu2023h,
  title={$H_{\infty}$-optimal observer design for systems with multiple delays in states, inputs and outputs: A {PIE} approach},
  author={Wu, Shuangshuang and Peet, Matthew M and Sun, Fuchun and Hua, Changchun},
  journal={International Journal of Robust and Nonlinear Control},
  volume={33},
  number={8},
  pages={4523--4540},
  year={2023},
  publisher={Wiley Online Library}
}

@inproceedings{braghini2025static,
  title={Static Output Feedback Stabilization of Linear Systems with Multiple Delays},
  author={Braghini, Danilo and Tognetti, Eduardo S and Peet, Matthew M},
  booktitle={2025 IEEE 64th Conference on Decision and Control (CDC)},
  pages={8365--8370},
  year={2025},
  organization={IEEE}
}

@inproceedings{jagt2023pie,
  title={A {PIE} representation of scalar quadratic {PDEs} and global stability analysis using {SDP}},
  author={Jagt, Declan and Seiler, Peter and Peet, Matthew},
  booktitle={2023 62nd IEEE Conference on Decision and Control (CDC)},
  pages={2950--2957},
  year={2023},
  organization={IEEE}
}

@article{peet2021partial,
  title={A partial integral equation ({PIE}) representation of coupled linear {PDEs} and scalable stability analysis using {LMIs}},
  author={Peet, Matthew M},
  journal={Automatica},
  volume={125},
  pages={109473},
  year={2021},
  publisher={Elsevier}
}

@article{jagt2025state,
  title={A State-Space Representation of Coupled Linear Multivariate {PDEs} and Stability Analysis using {SDP}},
  author={Jagt, Declan S and Peet, Matthew M},
  journal={arXiv preprint arXiv:2508.14840},
  year={2025}
}

@book{logan2004applied,
  title={Applied partial differential equations},
  author={Logan, John David and Logan, David},
  year={2004},
  publisher={Springer}
}

\end{document}